\documentclass[11pt]{article} 

\usepackage{fullpage}

\usepackage{amsmath,amsthm, amssymb}
\usepackage{latexsym}
\usepackage{graphicx}
\usepackage{ifthen}
\usepackage{subcaption}
\usepackage{multicol}
\usepackage{algorithm,algorithmic}
\usepackage[numbers]{natbib}
\usepackage{mathtools}
\usepackage{dsfont}
\usepackage{nicefrac}

\usepackage{wrapfig}
\usepackage{graphicx}

\usepackage{hyperref}
\hypersetup{colorlinks=true,linkcolor=blue,citecolor=blue,urlcolor=blue}

\newtheorem*{theorem*}{Theorem}
\newtheorem*{definition*}{Definition}

\newtheorem{theorem}{Theorem}
\newtheorem{lemma}{Lemma}

\usepackage{thm-restate}

\renewcommand{\algorithmiccomment}[1]{\bgroup\hfill\footnotesize~#1\egroup}
\usepackage{hyperref}
\hypersetup{colorlinks=true,linkcolor=blue,citecolor=blue,urlcolor=blue}

\newcommand{\algone}{{\textsc{Blits}}}
\newcommand{\algtwo}{{\textsc{Blits+}}}

\newcommand{\filter}{{\textsc{Sieve}}}

\newcommand{\INPUT}{\item[{\bf Input:}]}

\DeclareMathOperator{\E}{\mathbb{E}}
\newcommand{\EU}[1][]{\underset{#1}{\E}}

\DeclareMathOperator{\argmax}{argmax}

\newcommand{\OPT}{\texttt{OPT}}

\newcommand{\R}{\mathbb{R}}                     % Reals.
\newcommand{\U}{\mathcal{U}}

\newcommand{\D}{\mathcal{D}}
\renewcommand{\O}{\mathcal{O}}

\title{Non-monotone Submodular Maximization\\ in Exponentially Fewer Iterations}

\author{Eric Balkanski\\Harvard University \\ ericbalkanski@g.harvard.edu 
\and Adam Breuer \\ Harvard University \\ breuer@g.harvard.edu
\and Yaron Singer\\ Harvard University \\ yaron@seas.harvard.edu\\}

\date{}

\begin{document}

\setcounter{page}{0}

\maketitle
\begin{abstract}
In this paper we consider parallelization for applications whose objective can be expressed as maximizing a non-monotone submodular function under a cardinality constraint.
Our main result is an algorithm whose approximation is
arbitrarily close to $1/2e$ in $\O(\log^2 n)$ adaptive rounds, where $n$ is the size of the ground set. 
This is an exponential speedup in parallel running time over any previously studied algorithm for constrained non-monotone submodular maximization.
Beyond its provable guarantees, the algorithm performs well in practice.  Specifically, experiments on traffic monitoring and personalized data summarization applications show that the algorithm finds solutions whose values are competitive with state-of-the-art algorithms while running in exponentially fewer parallel iterations. 
\end{abstract}

\newpage

\section{Introduction}
In machine learning, many fundamental quantities we care to optimize such as entropy, graph cuts, diversity, coverage, diffusion, and clustering are submodular functions.
Although there has been a great deal of work in machine learning on applications that require constrained \emph{monotone} submodular maximization, many interesting submodular objectives are non-monotone.  Constrained non-monotone submodular maximization is used in large-scale personalized data summarization applications such as image summarization, movie recommendation, and revenue maximization in social networks~\cite{mirzasoleiman2016fast}.  In addition, many data mining applications on networks require solving constrained max-cut problems (see Section~\ref{sec:experiments}).

Non-monotone submodular maximization is well-studied~\cite{feige2011maximizing,lee2009non,gupta2010constrained, feldman2011unified, gharan2011submodular, Buch14, chekuri2015multiplicative, mirzasoleiman2016fast,ene2016constrained}, particularly under a cardinality constraint \cite{lee2009non,gupta2010constrained, gharan2011submodular, Buch14, mirzasoleiman2016fast}. For maximizing a non-monotone submodular function under a cardinality constraint $k$, a simple randomized greedy algorithm that iteratively includes a random element from the set of $k$ elements with largest marginal contribution at every iteration achieves a $1/e$ approximation to the optimal set of size $k$~\cite{Buch14}.  For more general constraints, Mirzasoleiman et al. develop an algorithm with strong approximation guarantees that works well in practice~\citep{mirzasoleiman2016fast}.  
While the algorithms for constrained non-monotone submodular maximization achieve strong approximation guarantees, their parallel runtime is linear in the size of the data due to their high \emph{adaptivity}. Informally, the adaptivity of an algorithm is the number of sequential rounds it requires when polynomially-many function evaluations can be executed in parallel in each round.  The adaptivity of the randomized greedy algorithm is $k$ since it sequentially adds elements in $k$ rounds.  The algorithm in Mirzasoleiman et al. is also $k$-adaptive, as is any known constant approximation algorithm for constrained non-monotone submodular maximization.  In general, $k$ may be $\Omega(n)$, and hence the adaptivity as well as the parallel runtime of all known constant approximation algorithms for constrained submodular maximization are at least \emph{linear} in the size of the data. 

%Adaptivity is crucial for parallelization since it lower bounds the parallel running time of an algorithm.  Specifically, the parallel running time of an $m$-adaptive algorithm is at least $\O(m)$.  
%For large-scale applications we therefore 
%For large-scale applications we seek algorithms with \emph{low} adaptivity, since these lead to algorithms that can be parallelized efficiently (see Appendix~\ref{sec:pram} for further discussion).   For this reason, adaptivity is studied across a wide variety of applications that include active learning \cite{Val75, Col88, BMW16},  communication complexity \cite{papadimitriou1984communication, duris1984lower, nisan1991rounds},  multi-armed bandits \cite{AAAK17}, sparse recovery \cite{HNC09,IPW11, haupt2009compressive}, and property testing \cite{CG17, buhrman2012non,chen2017settling}.  For submodular maximization, somewhat surprisingly, until very recently $\Omega(n)$ was the best known adaptivity (and hence best parallel running time) required for a constant factor approximation to monotone submodular maximization under a cardinality constraint.    

For large-scale applications we seek algorithms with \emph{low} adaptivity. Low adaptivity is what enables algorithms to be efficiently parallelized (see Appendix~\ref{sec:pram} for further discussion).   For this reason, adaptivity is studied across a wide variety of areas including online learning \cite{namkoong2017adaptive}, ranking \cite{Val75, Col88, BMW16}, multi-armed bandits \cite{AAAK17}, sparse recovery \cite{HNC09,IPW11, haupt2009compressive}, learning theory  \cite{CG17, buhrman2012non,chen2017settling}, and communication complexity  \cite{papadimitriou1984communication, duris1984lower, nisan1991rounds}.  For submodular maximization, somewhat surprisingly, until very recently $\Omega(n)$ was the best known adaptivity (and hence best parallel running time) required for a constant factor approximation to monotone submodular maximization under a cardinality constraint.

A recent line of work introduces new techniques for maximizing \emph{monotone} submodular functions under a cardinality constraint that produce algorithms that are $\mathcal{O}(\log n)$-adaptive and achieve strong constant factor approximation guarantees~\cite{BS18,BS18b} and even optimal approximation guarantees in  $\O(\log n)$ rounds~\cite{BRS18,EN18}.  This is tight in the sense that no algorithm can achieve a constant factor approximation with $\tilde{o}(\log n)$ rounds~\cite{BS18}.  Unfortunately, these techniques are only applicable to monotone submodular maximization and can be arbitrarily bad in the non-monotone case.
\begin{center}
\emph{
Is it possible to design fast parallel algorithms for non-monotone submodular maximization?}
\end{center}
%
%In this paper we design algorithms for constrained non-monotone submodular maximization with strong provable guarantees that perform well in theory and in practice.  
For unconstrained non-monotone submodular maximization, one can trivially obtain an approximation of $1/4$ in $0$ rounds by simply selecting a set uniformly at random~\cite{feige2011maximizing}.  We therefore focus on the problem of maximizing a non-monotone submodular function under a cardinality constraint.  

\paragraph{Main result.}  Our main result is the \algone \ algorithm, which obtains an approximation ratio arbitrarily close to $1/2e$ for maximizing a non-monotone (or monotone) submodular function under a cardinality constraint in $\O(\log^2 n)$ adaptive rounds (and $\O(\log^3 n)$ parallel runtime --- see Appendix~\ref{sec:pram}), where $n$ is the size of the ground set.  Although its approximation ratio is about half of the best known approximation for this problem~\cite{Buch14}, it achieves its guarantee in exponentially fewer rounds.  Furthermore, we observe across a variety of experiments that despite this slightly weaker worst-case guarantee, \algone \ consistently returns solutions that are competitive with the state-of-the-art.

\paragraph{Technical overview.} Non-monotone  submodular functions are  notoriously challenging  to optimize. Unlike in the monotone case,  standard algorithms  for submodular maximization such as the greedy algorithm perform arbitrarily poorly on non-monotone functions, and the best achievable approximation remains unknown.\footnote{To date, the best upper and lower bounds are ~\cite{Buch14} and ~\cite{gharan2011submodular} respectively for non-monotone submodular maximization under a cardinality constraint.} %Since the marginal contribution of an element to a set is not guaranteed to be non-negative, local decisions in the early stages of an algorithm potentially contribute negatively to the final value obtained by the algorithm. 
Since the marginal contribution of an element to a set is not guaranteed to be non-negative, an algorithm's local decisions in the early stages of optimization may contribute negatively to the value of its final solution. %Therefore, the high level algorithmic approach we use, which is to iteratively add to the solution blocks of elements obtained after aggressively discarding other elements, requires multiple subtle components for the analysis to go through for non-monotone functions. Specifically, we require that at every iteration, any element is added to the solution with low probability,  which imposes a significant additional challenge to just finding a block of high contribution at every iteration. 
At a high level, we overcome this problem with an algorithmic approach that iteratively adds to the solution blocks of elements obtained after aggressively discarding other elements. Showing the guarantees for this algorithm on non-monotone functions requires multiple subtle components. Specifically, we require that at every iteration, any element is added to the solution with low probability.  This requirement imposes a significant additional challenge to just finding a block of high contribution at every iteration, but it is needed to show that in future iterations there will exist a block with large contribution to the solution.
%
      %In particular, multiple subtle components are required to show the guarantees for the algorithmic approach we use, which is to iteratively add to the solution blocks of elements obtained by aggressively discarding elements. Specifically, we require that at every iteration, any element is added to the solution with low probability,  which imposes a significant additional challenge to just finding a block of high contribution at every iteration.
Second, we introduce a pre-processing step that discards elements with negative expected marginal contribution to a random set drawn from some distribution. This pre-processing step is needed for two different arguments: the first is that a large number of elements are discarded at every iteration, and the second is that a random block has high value when there are $k$ surviving elements.

\paragraph{Paper organization.}  Following a few preliminaries, we present the algorithm and its analysis in sections~\ref{sec:algorithm1} and~\ref{sec:algorithm2}.  We present the experiments in Section~\ref{sec:experiments}.

\paragraph{Preliminaries.}

 A function $f : 2^N \rightarrow \R_+$ is \emph{submodular} if the marginal contributions  $f_S(a) := f(S \cup a) - f(S)$ of an element $a \in N $ to a set $S \subseteq N$ are diminishing, i.e., $f_S(a) \geq f_T(a)$ for all   $a \in N \setminus T$ and  $S \subseteq T$. It is monotone if $f(S) \leq f(T)$ for all $S \subseteq T$. We assume that $f$ is non-negative, i.e., $f(S) \geq 0$ for all $S \subseteq N$, which is standard. %A submodular function $f$ is also   subadditive, meaning  $f(S \cup T) \leq f(S) + f(T)$ for all $S \subseteq T$.  
 We denote  the optimal solution  by $O$, i.e. $O := \argmax_{|S| \leq k}f(S)$, and its value by $\OPT := f(O)$.  We use the following lemma from \cite{feige2011maximizing}, which is useful  for non-monotone  functions:

\begin{lemma}[\cite{feige2011maximizing}]
\label{lem:random}
Let $g : 2^N \rightarrow \R$ be a non-negative submodular function. Denote by $A(p)$ a random subset of $A$ where each element appears with probability at most $p$ (not necessarily independently). Then,  $\E\left[g(A(p))\right] \geq (1-p)  g(\emptyset) + p \cdot g(A) \geq (1 - p)  g(\emptyset)$.
\end{lemma}

%Informally, the \emph{adaptivity} of an algorithm is the number of sequential rounds of queries it makes, where every round allows for polynomially-many parallel queries.  Formally, given a function $f$, an algorithm  is $r$-adaptive if every query $f(S)$ for the value of a set $S$  occurs at a round $i \in [r]$ such that $S$ is independent of the values $f(S')$ of all other queries  at round $i$.

\paragraph{Adaptivity.} Informally, the adaptivity of an algorithm is the number of sequential rounds it requires when polynomially-many function evaluations can be executed in parallel in each round. Formally, given a function $f$, an algorithm  is $r$-adaptive if every query $f(S)$ for the value of a set $S$  occurs at a round $i \in [r]$ such that $S$ is independent of the values $f(S')$ of all other queries  at round $i$.

\section{The \algone \ Algorithm}
\label{sec:algorithm1}

In this section, we describe the BLock ITeration Submodular maximization algorithm (henceforth \algone), which obtains an approximation arbitrarily close to  $1/2e$ in $\O(\log^2 n)$ adaptive rounds. %For simplicity of exposition, we analyze this algorithm and defer a discussion of the $\O(\log n)$-adaptive algorithm to Section~\ref{sec:logn}. 
\algone \ iteratively identifies a block  of at most $k/r$ elements using a \filter \ subroutine,  treated as a black-box in this section, and adds this block to the current solution $S$, for $r$ iterations. 
 
\begin{algorithm}[H]
\caption{\algone: the BLock ITeration Submodular maximization algorithm}
\begin{algorithmic}
    	\INPUT  constraint $k$, bound on number of iterations $r$
    	\STATE  $S \leftarrow \emptyset $
    	\STATE \textbf{for} $r$ \text{iterations}  $i = 1$ to $r$ \textbf{do}
	\STATE \hspace{\algorithmicindent} $S \leftarrow S \cup \filter(S, k, i, r)$
    	\RETURN $S$ 
  \end{algorithmic}
  \label{alg:1}
\end{algorithm}

The main challenge is to find in logarithmically many rounds a block of size at most $k/r$ to add to the current solution $S$. Before describing and analyzing the \filter \ subroutine, in the following lemma we reduce the problem of showing that \algone \ obtains a solution of value $\alpha v^{\star} /e$  to showing that \filter \ finds a block with marginal contribution at least $(\alpha/r) ((1 - 1/r)^{i-1}v^{\star} - f(S_{i-1}))$ to $S$ at every iteration $i$, where we wish to obtain $v^{\star}$  close to $\OPT$. The proof generalizes an argument in \cite{Buch14} and is deferred to Appendix~\ref{sec:appalgorithm1}.

\begin{restatable}{rLem}{lemmeta}
\label{lem:meta} 
For any $\alpha \in [0, 1]$,
assume that at iteration $i$ with current solution $S_{i-1}$, \filter \ returns a random set $T_i$ such that $$\E\left[f_{S_{i-1}}(T_i)\right] \geq \frac{\alpha}{r} \left(\left(1 - \frac{1}{r}\right)^{i-1}v^{\star} - f(S_{i-1})\right).$$ Then,
$$\E\left[f(S_r)\right] \geq\frac{\alpha}{e} \cdot  v^{\star}.$$
\end{restatable}
The advantage of \algone \ is that it terminates after $\O(d  \cdot \log n)$ adaptive rounds when using $r = \O(\log n)$ and a \filter \  subroutine that is $d$-adaptive.  In the next section we describe \filter \ and prove that it respects the conditions of Lemma~\ref{lem:meta} in $d = \O(\log n)$ rounds. 

\section{The \filter \ Subroutine}
\label{sec:algorithm2}
 In this section, we describe and analyze the  \filter \ subroutine. We show that for any constant $\epsilon > 0$,  this algorithm finds in $\O(\log n)$ rounds a block of at most $k/r$ elements with marginal contribution to $S$ that is at least $t/r$, with $t := ((1-\epsilon/2)/2)((1 - 1/r)^{i-1}(1 - \epsilon/2)\OPT - f(S_{i-1}))$, when called at iteration $i$ of \algone. By Lemma~\ref{lem:meta}  with $\alpha = (1-\epsilon)/2$ and $v^{\star} = (1 - \epsilon/2)\OPT$, this implies that \algone \ obtains an approximation arbitrarily close to $1/2e$ in $\O(\log^2 n)$ rounds. 
 
 The \filter \ algorithm, described formally below, iteratively discards elements from a set  $X$ initialized to the ground set $N$. We  denote by $\U(X)$ the uniform distribution over all subsets of $X$ of size exactly $k/r$ and by $\Delta(a, S, X)$  the expected marginal contribution of an element $a$ to a union of the current solution $S$ and a random set $R \sim \U(X)$, i.e. 
$$\Delta(a, S, X) := \E_{R \sim \U(X)} \left[f_{S \cup (R\setminus a)}(a) \right].$$
 %
 %Note that if $\E_{R \sim \U(X)}\left[f_S(R)\right] \geq t/r$, then \filter \ has found a random block with the desired expected marginal contribution to $S$. 
%At every iteration, \filter \ first pre-processes surviving elements $X$ to obtain $X^+$, which is the set of elements $a \in X$ with non-negative marginal contribution $\Delta(a, S, X)$. This pre-processing step is then used to  evaluate the marginal contribution $\E_{R \sim \U(X)}\left[f_S(R \cap X^+)\right]$ of a random set $R \sim \U(X)$ without its elements that have negative marginal contribution. If the marginal contribution of $R \cap X^+$ is at least $t/r$, then $R \cap X^+$ is returned. Otherwise, the algorithm discards from $X$
% the elements $a$ with expected marginal contribution $\Delta(a, S, X)$  that is smaller than $(1+\epsilon/2)t /k$.
% The algorithm iterates until either $\E[f_S(R \cap X^+)] \geq t/r$ or there are less than $k$ surviving elements, in which case \filter \  returns a random set $R \cap X^+$ with $R \sim \U(X)$ and with dummy elements  added to $X$ so that $|X| = (1+\epsilon) k$. A dummy element $a$ for which $f_S(a) = 0$ for all $S$.

At every iteration, \filter \ first pre-processes surviving elements $X$ to obtain $X^+$, which is the set of elements $a \in X$ with non-negative marginal contribution $\Delta(a, S, X)$. After this pre-processing step, \filter \  evaluates the marginal contribution $\E_{R \sim \U(X)}\left[f_S(R \cap X^+)\right]$ of a random set $R \sim \U(X)$ without its elements  not in $X^+$ (i.e. $R$ excluding its elements with negative expected  marginal contribution). If the marginal contribution of $R \cap X^+$ is at least $t/r$, then $R \cap X^+$ is returned. Otherwise, the algorithm discards from $X$
 the elements $a$ with expected marginal contribution $\Delta(a, S, X)$ less than $(1+\epsilon/2)t /k$.
 The algorithm iterates until either $\E[f_S(R \cap X^+)] \geq t/r$ or there are less than $k$ surviving elements, in which case \filter \  returns a random set $R \cap X^+$ with $R \sim \U(X)$ and with dummy elements  added to $X$ so that $|X| = k$. A dummy element $a$ is an element with $f_S(a) = 0$ for all $S$.  

%Three difficulties arise for non-monotone functions, these are further discussed in their corresponding part in the analysis. The first is  that the blocks of elements added to  $S$ need to be randomized (Section~\ref{sec:}). The second is that at some iterations of \filter, a random set has low expected value and there are only a small number of elements with low marginal contribution (Section~\ref{sec:}). The third is to guarantee that the elements that survive \filter \ maintain sufficient value when there are only a small number of elements remaining (Section~\ref{sec:}).

%where $S(1/2)$ is the uniformly random subset of $S$ where each element in $S$ is sampled independently with probability $1/2$
 \begin{algorithm}[H]
\caption{\filter$(S, k, i, r)$}
\begin{algorithmic}
    	\INPUT current solution $S$ at outer-iteration $ i \leq r$
    	\STATE $X \leftarrow N, t \leftarrow \frac{1-\epsilon/2}{2} (\left(1 - 1/r\right)^{i-1}(1 - \epsilon/2)\OPT - f(S))$
   	\STATE  \textbf{while}   $|X| > k$  \textbf{do} 
   	\STATE \ \ \ \hspace{\algorithmicindent} $X^+ \leftarrow \left\{a \in X \ : \ \Delta(a, S, X) \geq 0 \right\}$
   	\STATE \ \ \ \hspace{\algorithmicindent}  \textbf{if} $\E_{R \sim \U(X)}\left[ f_S(R 
   	\cap X^+)\right]  \geq  t/r$  \textbf{return} $R \cap X^+$, where $R \sim \U(X)$ 
   			\STATE  \ \ \ \hspace{\algorithmicindent}  $X \leftarrow  \left\{a \in X \ : \ \Delta(a, S, X) \geq   \left(1 + \epsilon/4 \right) t /k\right\} $	
    		\STATE $X \leftarrow X \cup \{k - |X| \text{ dummy elements}\}$
    		\STATE  $X^+ \leftarrow \left\{a \in X \ : \ \Delta(a, S, X) \geq 0 \right\}$
   			\STATE  \textbf{return} $R \cap X^+$, where $R \sim \U(X)$ 
  \end{algorithmic}
  \label{alg:filter}
\end{algorithm}

The above description is an idealized version of the algorithm. In practice, we do not know $\OPT$ and we cannot compute expectations exactly. Fortunately, we can apply multiple guesses for  $\OPT$ non-adaptively  and obtain arbitrarily good estimates of the expectations in one  round by sampling. The sampling process for the estimates first samples $m$ sets from $\U(X)$, then queries the desired sets to obtain a random realization of $f_S(R \cap X^+)$ and $f_{S \cup (R\setminus a)}(a)$, and finally averages the $m$ random realizations  of these values. By standard concentration bounds, $m = \O((\OPT/\epsilon)^2 \log(1/\delta))$ samples are sufficient to obtain with probability $1 - \delta$ an estimate with an $\epsilon$ error.  For ease of presentation and notation, we analyze the idealized version of the algorithm, which easily extends to the algorithm with estimates and guesses as in \cite{BS18, BS18b, BRS18}.

\subsection{The approximation}
Our goal is to show that \filter \ returns a random block whose expected marginal contribution to $S$ is at least  $t/r$. 
By Lemma~\ref{lem:meta} this implies  \algone \ obtains a $(1-\epsilon)/2e$-approximation. 

\begin{lemma}
\label{lem:main} Assume  $r \geq 20 \rho \epsilon^{-1}$ and that after at most $\rho - 1$ iterations of \filter,
 \filter \ returns a set $R$ at iteration $i$ of \algone, then
$$\E[f_S(R)] \geq  \frac{t}{r} =  \frac{1-\epsilon/2}{2r}\left(\left(1 - \frac{1}{r}\right)^{i-1} (1-\epsilon/2)\OPT - f(S)\right).$$
\end{lemma}
The remainder of the analysis of the approximation is devoted to the proof of Lemma~\ref{lem:main}. First note that if \filter \ returns $R \cap X^+$, then the desired bound  on $\E[f_S(R)]$ follows from the condition to return that block. Otherwise \filter \ returns $R$ due to $|X| \leq k$, and then the proof consists of two parts.  First, in Section \ref{sec:part1} we argue that when \filter \ terminates,  there \emph{exists} a subset $T$ of $X$ for which $f_S(T) \geq t$. Then, in Section~\ref{sec:part2} we prove that such a subset $T$ of $X$ for which $f_S(T) \geq t$ not only exists, but is also returned by \filter .  We do this by proving a new general lemma for non-monotone submodular functions that may be of independent interest.  This lemma shows that a random subset of $X$ of size $s$ well approximates the optimal subset of size $s$ in $X$.  
%This lemma gives an approximation guarantee when choosing a subset of size $s$ from $X$   
%This lemma provides an approximation guarantee of a subset $X$ of size $s$ 
%This lemma provides an approximation guarantee when returning a uniformly random subset of $X$ of size $s$, in terms of the value of the optimal subset of $X$.
 
\subsubsection{Existence of a surviving block with high contribution to $S$}
\label{sec:part1}

The main result in this section is Lemma~\ref{lem:Tstar}, which shows that when \filter \ terminates there \emph{exists} a subset $T$ of $X$ s.t $f_S(T) \geq t$. %This set $T$ is a subset of the optimal solution $O$.  Lemma~\ref{lemTstar} relies on showing that meaningful elements in the  optimal solution survive filtering. This argument uses Lemma~\ref{lem:optContrib} which shows that the marginal contribution  of $O$ to $S$ and, for each $j$, random set $R_j \sim \U(X)$ at iteration $j$ of \filter approximates $t$.
To prove this, we first prove Lemma~\ref{lem:margS}, which argues that $f(O \cup S) \geq (1 - 1/r)^{i-1} \OPT $. This bound explains the $(1 - 1/r)^{i-1}(1 - \epsilon/2)\OPT - f(S_{i-1}))$ term in $t$. For monotone functions, this is trivial since $f(O \cup S)  \geq f(O) = \OPT$ by definition of monotonicity.
%
%%zzz%% For non-monotone functions, this inequality does not hold. Instead, the main approach used to bound  $f(O \cup S)$ is to argue that any element $a \in N$ is added to $S$ by \filter \ with probability at most $1- 1/r$ at every iteration. The crucial component of the algorithm for that argument is that in both cases where \filter \ terminates we have that $|X| > k$. 
For non-monotone functions, this inequality does not hold. Instead, the approach used to bound  $f(O \cup S)$ is to argue that any element $a \in N$ is added to $S$ by \filter \ with probability at most $1/r$ at every iteration. The key to that argument is that in both cases where \filter \ terminates we have $|X| \geq k$ (with $X$ possibly containing dummy elements), which implies that every element $a$ is in $R \sim \U(X)$ with probability at most $1/r$.

\begin{lemma}
\label{lem:margS}
Let $S$ be the set obtained after $i-1$ iterations of \algone  \ calling the  \filter \ subroutine, then  $$\E[f(O \cup S)] \geq (1 - 1/r)^{i-1} \OPT.$$ 
\end{lemma}
\begin{proof} 
In both cases where \filter  \ terminates, $|X| \geq k$. Thus $\Pr[a \in R \sim \U(X)] = k/(r|X|) < 1/r$. This implies that at iteration $i$ of \algone, $\Pr[a \in S] \leq 1 - (1 - 1/r)^{i-1}$. Next, we define $g(T) := f(O \cup T)$, which is also submodular. By Lemma~\ref{lem:random} from the preliminaries, we get 
\begin{equation*}
\E[f(S \cup O)] = \E[g(S)] \geq (1 - 1/r)^{i-1}g(\emptyset) =  (1 - 1/r)^{i-1} \OPT . \qedhere  
\end{equation*}
%
%We conclude that $\E[f_{S}(O)] \geq (1 - 1/r)^{i-1} \OPT - f(S).$ 
\end{proof}

Let $\rho$, $X_j$, and $R_j$ denote the number of iterations of $\filter(S, k, i ,r)$, the set $X$ at iteration $j \leq \rho$ of \filter, and the set $R \sim \U(X_j)$ respectively. We show that the expected marginal contribution of $O$ to $S \cup \left(\cup_{j=1}^\rho R_j \right)$ approximates $(1 - 1/r)^{i-1} \OPT - f(S)$ well. This crucial fact allows us to argue about the value of optimal elements that survive iterations of \filter.

 \begin{restatable}{rLem}{optContrib} 
\label{lem:optContrib}
For all  $r, \rho, \epsilon > 0$ s.t. $r \geq 20 \rho \epsilon^{-1}$, if \filter$(S, k, i, r)$ has not terminated after $\rho$ iterations, then
$$\E_{R_1, \ldots, R_\rho}\left[f_{S \cup \left(\cup_{j=1}^\rho R_j \right)}\left(O\right)\right] \geq \left(1 - \epsilon/10\right)\left((1 - 1/r)^{i-1} (1 - \epsilon/2)\OPT - f(S)\right).$$
\end{restatable}
\begin{proof} 
We exploit the fact that if \filter$(S, k, i, r)$  has not terminated after $\rho$ iterations, then by the algorithm, the random set $R_j \sim \U(X)$ at iteration $j$ of \filter \ has expected value that is upper bounded as follows:
$$\EU[R_j]\left[f_{S }\left(  R_j \right)\right] < \frac{1-\epsilon/2}{2r}  \left(\left(1 - \frac{1}{r}\right)^{i-1}(1 - \epsilon/2)\OPT - f(S)\right)$$
for all $j \leq \rho$. Next, by subadditivity, we have  
$\E_{R_1, \ldots, R_\rho}\left[f_{S }\left( \left(\cup_{j=1}^\rho R_j \right)\right)\right]  \leq \sum_{j=1}^\rho \E_{R_j}\left[f_{S }\left(  R_j \right)\right]$.

Note that $$\Pr_{R_j}\left[a \in R_j\right] \leq \frac{k/r}{|X_j|} \leq \frac{1}{r}$$
since $|X| > k$ during \filter.
Thus, by a union bound, $$ \Pr_{R_1, \ldots, R_\rho}\left[a \in \cup_{j=1}^\rho R_j\right]  \leq \frac{\rho}{r} \leq \frac{\epsilon}{20}.$$
Next, define $$g(T) = f(O \cup S \cup T)$$ which is non-negative submodular.
Thus,  we have
\begin{align*}
\EU[R_1, \ldots, R_\rho]\left[f_{S }\left(O\cup \left(\cup_{j=1}^\rho R_j \right)\right)\right] & = \EU[R_1, \ldots, R_\rho]\left[g\left(\cup_{j=1}^\rho R_j \right)\right] - f(S) \\
& \geq \left(1 - \frac{\epsilon}{20}\right)g(\emptyset) - f(S) & \text{Lemma~\ref{lem:random}}\\ 
& = \left(1 - \frac{\epsilon}{20}\right)f(O \cup S) - f(S) \\
& \geq \left(1 - \frac{\epsilon}{20}\right) (1 - 1/r)^{i-1} \OPT - f(S) &\text{Lemma~\ref{lem:margS}}
\end{align*}

 Combining the above inequalities, we conclude that
\begin{align*}
&  \EU[R_1, \ldots, R_\rho]\left[f_{S \cup \left(\cup_{j=1}^\rho R_j \right)}\left(O\right)\right] \\
 = & \EU[R_1, \ldots, R_\rho]\left[f_{S }\left(O\cup \left(\cup_{j=1}^\rho R_j \right)\right)\right] -  \EU[R_1, \ldots, R_\rho]\left[f_{S }\left( \left(\cup_{j=1}^\rho R_j \right)\right)\right] \\
  \geq  &  \left( \left(1 - \frac{\epsilon}{20}\right) (1 - 1/r)^{i-1} \OPT - f(S)\right) - \sum_{j=1}^\rho \EU[R_j]\left[f_{S }\left(R_j \right)\right]\\
   \geq  & \left( \left(1 - \frac{\epsilon}{20}\right) (1 - 1/r)^{i-1} \OPT - f(S)\right)- \frac{\rho(1-\epsilon/2)}{2r}  \left(\left(1 - \frac{1}{r}\right)^{i-1} (1 - \epsilon/2)\OPT - f(S)\right) \\ 
    \geq   & \left(1 - \frac{\epsilon}{10}\right)\left((1 - 1/r)^{i-1} (1 - \epsilon/2)\OPT - f(S)\right)    \qedhere
\end{align*} 
\end{proof}

We are now ready to show that when \filter \ terminates after $\rho$ iterations, there exists a subset $T$ of $X_\rho$ s.t $f_S(T) \geq t$. At a high level, the proof defines $T$ to be a set of meaningful optimal elements, then uses Lemma~\ref{lem:optContrib} to show that these elements survive $\rho$  iterations of \filter \ and respect $f_S(T) \geq t$.

\begin{restatable}{rLem}{lemTstar}
\label{lem:Tstar} 
For all $r, \rho, \epsilon > 0$, if $r \geq 20 \rho \epsilon^{-1}$, then there exists  $T \subseteq X_\rho$, that survives $\rho$ iterations of \filter$(S, k, i, r)$ and that satisfies 
$$f_S(T) \geq  \frac{1 - \epsilon/10}{2} \left((1 - 1/r)^{i-1} (1- \epsilon/2)\OPT - f(S)\right).$$
\end{restatable}
\begin{proof}
At a high level, the proof first defines a subset $T$ of the optimal solution $O$. Then, the remainder of the proof consists of two main parts. First, we show that elements in $T$ survive $\rho$ iterations of \filter$(S, k, i, r)$. Then, we show that $f_S(T) \geq \frac{1}{2}\left(1 - \frac{\epsilon}{10}\right)\left((1 - 1/r)^{i-1} (1 - \epsilon/2)\OPT - f(S)\right) .$ We introduce some notation. Let $O = \{o_1, \ldots, o_k\}$ be the optimal elements in some arbitrary order and $O_\ell = \{o_1, \ldots, o_\ell\}$. We define the following marginal contribution $\Delta_\ell$ of  optimal element $o_\ell$:
$$\Delta_\ell := \EU[R_1, \ldots, R_\rho] \left[ f_{S \cup O_{\ell-1}\cup \left(\cup_{j=1}^\rho R_j \setminus \{o_\ell\}\right) }(o_\ell)\right].$$
We define $T$ to be the set of optimal elements $o_\ell$ such that $\Delta_\ell \geq \frac{1}{2} \Delta$ where
$$\Delta := \frac{1}{k}\cdot \EU[R_1, \ldots, R_\rho]\left[f_{S \cup \left(\cup_{j=1}^\rho R_j \right)}\left(O\right)\right].$$
 We first argue that elements in $T$ survive $\rho$ iterations of \filter$(S, k, i, r)$. For element $o_\ell \in T$, we have
$$\Delta_\ell  \geq   \frac{1}{2} \Delta  \geq \frac{1}{2k}  \cdot \EU[R_1, \ldots, R_\rho]\left[f_{S \cup \left(\cup_{j=1}^\rho R_j \right)}\left(O\right)\right]
\geq  \frac{1}{2k}  \left(1 - \frac{\epsilon}{10}\right)\left((1 - 1/r)^{i-1} (1 - \epsilon/2)\OPT - f(S)\right) $$
 Thus, at iteration $i \leq \rho$, by submodularity, 
\begin{align*}
\EU[R_j] \left[f_{S \cup (R_j \setminus \{o_\ell\})}(o_\ell) \right] &  \geq  \EU[R_1, \ldots, R_\rho] \left[ f_{S \cup O_{\ell-1}\cup \left(\cup_{j=1}^\rho R_j  \setminus \{o_\ell\}\right)}(o_\ell)\right] \\
& = \Delta_\ell \\
& \geq \frac{1}{2k}  \left(1 - \frac{\epsilon}{10}\right)\left((1 - 1/r)^{i-1} (1 - \epsilon/2)\OPT - f(S)\right) \\
&\geq \left(1 + \epsilon/4 \right)\frac{1 - \epsilon/2}{2k}  \left(\left(1 - \frac{1}{r}\right)^{i-1}\OPT - f(S)\right) 
\end{align*}
 and $o_\ell$ survives all iterations $j \leq \rho$, for all $o_\ell \in T$. Next, note that
$$\sum_{\ell =1}^k \Delta_\ell \geq   \EU[R_1, \ldots, R_\rho]\left[f_{S \cup \left(\cup_{j=1}^\rho R_j \right)}\left(O\right)\right]  =  k \Delta.$$

 Next, observe that $$ \sum_{\ell =1}^k \Delta_\ell = \sum_{o_\ell \in T} \Delta_\ell + \sum_{\ell  \in O \setminus T} \Delta_\ell \leq \sum_{o_\ell \in T} \Delta_\ell +  \frac{k}{2}\Delta.$$
By combining the two inequalities above, we get  $\sum_{o_\ell \in T} \Delta_\ell  \geq  \frac{k}{2}  \Delta$. Thus, by submodularity,
\begin{align*}
f_S(T)  \geq \sum_{o_\ell \in T} f_{S \cup O_{\ell-1}}\left(o_\ell\right) 
 \geq \sum_{o_\ell \in T}\EU[R_1, \ldots, R_\rho] \left[ f_{S \cup O_{\ell-1}\cup \left(\cup_{j=1}^\rho R_j \setminus \{o_\ell\}\right)}(o_\ell)\right]  
 = \sum_{o_\ell \in T} \Delta_\ell  
 \geq  \frac{k}{2} \Delta. 
\end{align*}
We conclude that 
\begin{align*}
f_S(T)  \geq  \frac{k\Delta }{2} &  =  \frac{1}{2} \EU[R_1, \ldots, R_\rho]\left[f_{S \cup \left(\cup_{j=1}^\rho R_j \right)}\left(O\right)\right]\\
&  \geq   \frac{1}{2}\left(1 - \frac{\epsilon}{10}\right)\left((1 -1/r)^{i-1} (1 - \epsilon/2)\OPT - f(S)\right) \qedhere
\end{align*}
\end{proof}

\subsubsection{A random subset approximates the best surviving block}
\label{sec:part2}
In the previous part of the analysis, we showed the existence of a surviving set $T$ with contribution at least $\frac{1 - \epsilon/10}{2} \left((1 - 1/r)^{i-1} (1-\epsilon/2)\OPT - f(S)\right)$ to $S$. In this part, we show that the random set $R \cap X^+$, with $R \sim \U(X)$, is a $1/r$ approximation to any surviving set $T \subseteq X^+$ when  $|X| = k$. A key component of the algorithm for this argument to hold for non-monotone functions is the final pre-processing step to restrict $X$ to $X^+$ after adding dummy elements. We use this restriction to argue that every element $a \in R \cap X^+$ must contribute a non-negative expected value to the set returned.

\begin{restatable}{rLem}{lemrandomtwo}
\label{lem:randomtwo}
Assume \filter  \ returns $R \cap X^+$ with $R \sim \U(X)$ and $|X| = k$.  For any $T \subseteq X^+$,  we have
$$\E_{R \sim \U(X)}[f_S(R \cap X^+)] \geq  f_S(T) /r.$$
\end{restatable}
\begin{proof}
Let $T \subseteq X^+$. First note that
\begin{align*}
\E[f_S(R \cap X^+)] & = \E[f_S((R \cap X^+) \cap T)] + \E[f_{S \cup ((R \cap X^+) \cap T)}((R \cap X^+) \setminus T)] \\
& =  \E[f_{S}(R  \cap T)] + \E[f_{S  \cup (R \cap T)}((R \cap X^+) \setminus T)].
\end{align*}
 where the second inequality is due to the fact that $T \subseteq X^+$. We first bound $\E[f_S(R \cap T)]$. Let $T = \{a_1, \ldots, a_\ell\}$ be some arbitrary ordering of the elements in $T$ and define $T_i = \{a_1, \ldots, a_i\}.$  Then, 
\begin{align*}
\E[f_S(R  \cap T)]  = \E\left[\sum_{a_i \in R \cap T}f_{S \cup (R \cap T_{i-1})}(a_i) \right]  & \geq  \E\left[\sum_{a_i \in R \cap T}f_{S \cup  T_{i-1}}(a_i) \right]& \text{submodularity} \\
& = \E\left[\sum_{a_i \in  T} \mathds{1}_{a_i \in R} \cdot f_{S \cup  T_{i-1}}(a_i) \right]&  \\
& = \sum_{a_i \in  T} f_{S \cup  T_{i-1}}(a_i) \cdot \E\left[ \mathds{1}_{a_i \in R} \right]&  \\
& = \frac{1}{r}  \sum_{a_i \in  T} f_{S \cup  T_{i-1}}(a_i) &  \\
& = \frac{1}{r}  f_S(T).&  
\end{align*}
where the third equality is due to the fact that  $\Pr[a_i \in R] = \frac{k/r}{|X|} = 1/r$.
Next, we bound $\E[f_S((R \cap X^+) \setminus T)]$. Similarly as in the previous case, assume that $X^+ = \{a_1, \ldots, a_\ell\}$ is some arbitrary ordering of the elements in $X^+$ and define $X^+_i = \{a_1, \ldots, a_i\}$. Observe that
\begin{align*}
\E[f_{S  \cup (R \cap T)}((R \cap X^+) \setminus T)] & = \E\left[\sum_{a_i \in (R \cap X^+) \setminus T} f_{S  \cup (R \cap T) \cup ((R \cap X_{i-1}^+) \setminus T)}(a_i)\right] \\
& \geq \E\left[\sum_{a_i \in (R \cap X^+) \setminus T} f_{S  \cup (R \setminus a_i)}(a_i)\right] & \text{submodularity} \\
& = \E\left[\sum_{a_i \in  X^+ \setminus T} \mathds{1}_{a_i \in R} \cdot  f_{S  \cup (R \setminus a_i)}(a_i)\right] &  \\
& = \sum_{a_i \in  X^+ \setminus T} \E\left[ \mathds{1}_{a_i \in R} \cdot  f_{S  \cup (R \setminus a_i)}(a_i)\right] &  \\
& = \sum_{a_i \in  X^+ \setminus T}  \Pr[a_i \in R]\cdot   \E\left[ f_{S  \cup (R \setminus a_i)}(a_i) | a_i \in R \right] &  \\
& \geq \sum_{a_i \in  X^+ \setminus T}  \Pr[a_i \in R]\cdot   \E\left[ f_{S  \cup (R \setminus a_i)}(a_i)  \right]   & \text{submodularity}  \\
& \geq \sum_{a_i \in  X^+ \setminus T}  \Pr[a_i \in R] \cdot 0 & a_i \in X^+ 
\end{align*}
We conclude that $\E[f_S(R \cap X^+)]  \geq \frac{1}{r}  f_S(T)$
\end{proof}

There is a tradeoff between the contribution $f_S(T)$ of the best surviving set $T$ and the contribution of a random set $R \cap X^+$ returned in the middle of an iteration due to the thresholds  $(1+\epsilon/4)t/k$ and $t/r$, which is controlled by $t$. The optimization of this tradeoff explains the $(1-\epsilon/2)/2$ term in $t$.

%This part of the analysis gives a new general lemma for non-monotone functions to show an approximation guarantee when returning a uniformly random subset of $X = A \cup B$ of size $s$, in terms of the value of a subset $A$ of $X$. This lemma resembles Lemma 2.3 from \cite{feige2011maximizing}, which cannot be applied in our case because it requires sampling $A$ and $B$ independently.
%
%\begin{restatable}{rLem}{lemrandomtwo}
%\label{lem:randomtwo}
%Let $f : 2^N \rightarrow \R^+$ be a non-negative submodular function, $A, B \subseteq N$ be disjoint sets, and $R \sim \U(A \cup B, s)$ be a uniformly random subset of $A \cup B$ of size $s$, then $\E[f(R)] \geq \left(1 - \frac{s}{|B|}\right) \frac{s}{|A \cup B|} f(A)$.
%\end{restatable}
%\begin{proof}[Proof Sketch (full proof in Appendix~\ref{sec:appalgorithm2})] 
%We apply Lemma~\ref{lem:random} twice. First, with $A \cap R = A'$,  we define $g(T) = f(A' \cup T)$ to get $\E[f(R)| A \cap R = A']   = \E[(1 - q) f(A \cap R) | A \cap R = A']$ where $q = \Pr[b \in R : b \in B, A \cap R = A']$. By bounding $q$, we can then uncondition to get $\E[f(R)] \geq \left(1 - \frac{s}{|B|}\right) \E[f(A \cap R)]$. We then apply Lemma~\ref{lem:random} a second time on $\E[f(A \cap R)]$ to conclude that $\E[f(R)] \geq \left(1 - \frac{s}{|B|}\right) \frac{s}{|A \cup B|}\cdot f(A)$.
%\end{proof}

\subsubsection{Proof of main lemma}

\begin{proof}[Proof of Lemma~\ref{lem:main}]
There are two cases. If \filter \ returns $R \cap X^+$ in the middle of an iteration, then by the condition to return that set, $\E_{R \sim \U(X)}\left[ f_S(R 
   	\cap X^+)\right]  \geq t/r = \frac{1 - \epsilon/2}{2} (\left(1 - 1/r\right)^{i-1}(1-\epsilon/2)\OPT - f(S))/r$.   	Otherwise, \filter  \ returns $R \cap X^+$ with $|X| = k$. By Lemma~\ref{lem:Tstar}, there exists $T \subseteq X_{\rho}$ that survives $\rho$ iterations of \filter \ s.t. $f_S(T) \geq  \frac{1 - \epsilon/10}{2} \left((1 - 1/r)^{i-1} (1-\epsilon/2)\OPT - f(S)\right).$ Since there are at most $\rho - 1$ iterations of \filter, $T$ survives each iteration and the final pre-processing. This implies that
     	$T \subseteq X^+$ when the algorithm terminates.  By Lemma~\ref{lem:randomtwo}, we then conclude that
 $\E_{R \sim \U(X)}[f_S(R \cap X^+)] \geq  f_S(T) /r \geq \frac{1 - \epsilon/10}{2r} \left((1 - 1/r)^{i-1} (1-\epsilon/2)\OPT - f(S)\right) \geq t/r$.
 %
 % $\E[f_S(R)] \geq \left(1 - \frac{s}{|B|}\right) \frac{s}{|A \cup B|}\cdot f(A) \geq  \left(1 - \frac{2}{\epsilon r}\right) \frac{1}{(1 +\epsilon/2) r}\cdot f_S(O_X)$ since $|B| \geq \epsilon k/2$. By Lemma~\ref{lem:Tstar}, since $|T| \leq k$, we have $f_S(O_X) \geq   \frac{1}{2}\left(1 - \frac{\epsilon}{20} - \frac{1}{2r}\right)\left((1 - 1/r)^{i-1} \OPT - f(S)\right)$ and we conclude that $\E[f_S(R)] \geq    \frac{1}{2} \left(1 - \frac{2}{\epsilon r}\right) \frac{1}{(1 +\epsilon/2) r}\cdot \left(1 - \frac{\epsilon}{20} - \frac{1}{2r}\right)\left((1 - 1/r)^{i-1} \OPT - f(S)\right) \geq \frac{1-\epsilon}{2r} \left((1 - 1/r)^{i-1} \OPT - f(S)\right)$ where the last inequality is with $r = 20 \epsilon^{-1} \log_{1+\epsilon/2}(n)$.
\end{proof}

\subsection{The adaptivity of \filter \ is $\O(\log n)$}
\label{sec:part3}

We now observe that the number of iterations of \filter \ is $\O(\log n)$. This logarithmic adaptivity is due to the fact that  \filter \ either returns a random set or discards a constant fraction of the surviving elements at every iteration. Similarly to Section~\ref{sec:part2}, the pre-processing step to obtain $X^+$ is crucial to argue that since a random subset $R \cap X^+$ has contribution below the $t/r$ threshold and since all elements in $X^+$ have non-negative marginal contributions, there exists a large set of elements in $X^+$ with expected marginal contribution to $S \cup R$ that is below the $(1+\epsilon/4)t/k$ threshold.

\begin{restatable}{rLem}{lempruning}
\label{lem:pruning}
Let $X_j$ and $X_{j+1}$ be the surviving elements $X$ at the start and end of iteration $j$ of \filter$(S, k, i, r)$. For all $S \subseteq N$ and $r, j, \epsilon > 0$, if \filter$(S, k, i, r)$ does not terminate at iteration $j$, then
$|X_{j+1}| < |X_j|/({1+\epsilon/4}).$
\end{restatable}
\begin{proof}
At a high level, since the surviving elements must have high value and a random set has low value, we can then use the thresholds to bound how many such surviving elements there can be while also having a random set of low value. To do so, we focus on the value of $f(R_j \cap X_{j+1})$ of the surviving elements $X_{j+1}$ in a random set $R_j \sim \D_{X_j}$.

 We denote by $\{a_1, \ldots, a_\ell\}$ the elements in $R_j \cap X_j^+ $.  Observe that
\begin{align*}
& \E \left[f_S(R_j \cap X^+_j)\right] \\
 = & \E\left[\sum_{j =1}^\ell f_{S \cup \{a_1, \ldots, a_{j-1}\}}(a_j)\right] & \\
 \geq  & \E\left[\sum_{a \in R_j \cap X_{j}^+} f_{S \cup (R_j  \setminus a)}(a)\right] & \text{submodularity}\\
 =  & \E\left[\sum_{a \in X_{j}^+} \mathds{1}_{a \in R_j} \cdot  f_{S \cup (R_j \setminus a)}(a)\right]  & \\
 =  & \sum_{a \in X_{j}^+} \E\left[\mathds{1}_{a \in R_j} \cdot  f_{S \cup (R_j \setminus a)}(a)\right]. & \\
  = & \sum_{a \in X_{j}^+} \Pr\left[a \in R_j\right] \cdot \E\left[ f_{S \cup (R_j \setminus a)}(a) | a \in R_j\right] & \\
 \geq & \sum_{a \in X_{j}^+} \Pr\left[a \in R_j\right] \cdot \E\left[ f_{S \cup (R_j \setminus a)}(a) \right] & \text{submodularity}
 \end{align*} 
 Then, for $a \in X_j^+$, either $a$ survives this iteration $i$ and $a \in X_{j+1}$ or it is filtered and $a \in X_j^+ \setminus X_{j+1}$. If $a \in X_{j+1}$, then $\E\left[ f_{S \cup (R_j \setminus a)}(a) \right]  \geq \left(1 + \epsilon/4 \right)t/k$ by the algorithm and $\Pr\left[a \in R_j\right] = \frac{k}{r|X_j|}$ by the definition of $\U(X)$. Thus,
$$\Pr\left[a \in R_j\right] \cdot \E\left[ f_{S \cup (R_j \setminus a)}(a) \right] \geq \frac{1}{r|X_j|} \cdot\left(1 + \epsilon/4 \right)t.$$
If $a \in X_j^+ \setminus X_{j+1}$, then $\E\left[ f_{S \cup (R_j \setminus a)}(a) \right]  \geq 0$ by the definition of $X_j^+$. Putting all the previous pieces together, we get 
$$\E \left[f_S(R_j \cap X^+_j)\right]  \geq |X_{j+1}| \cdot \frac{1}{r|X_j|} \cdot\left(1 + \epsilon/4 \right)t.$$
 Next, since elements are discarded, a random set must have low value by the algorithm,
$t/r \geq \E \left[f_S(R_j \cap X_j^+)\right].$
  Finally, by combining the above  inequalities, we  conclude that $|X_{j+1}| \leq |X_j|/ (1+\epsilon/4)$. 
\end{proof}

%Unlike for monotone functions, a random set of low expected value does not imply that there is a large number of elements with low expected marginal contribution to a random set since elements can have negative marginal contributions. The main idea is that for non-monotone functions, we use a ``pre-processing" step to only consider elements $X^+$ with positive marginal contributions. This  ensures that if a random set has low value, then a constant fraction of surviving elements are removed.

\subsection{Main result for \algone}

%We are now ready to show the main result for \algone. 

\begin{theorem}  For any constant $\epsilon > 0$, \algone  \ initialized with $r = 20 \epsilon^{-1} \log_{1+\epsilon/2}(n) $ is  $\mathcal{O}\left(\log^2 n\right)$-adaptive and obtains a $\frac{1  - \epsilon}{2e}$ approximation.
\end{theorem}
\begin{proof}
By Lemma~\ref{lem:main}, we have $\E[f_S(R)] \geq  \frac{1-\epsilon/2}{2}\left(\left(1 - \frac{1}{r}\right)^{i-1} (1-\epsilon/2)\OPT - f(S)\right).$ Thus, by Lemma~\ref{lem:meta} with $\alpha = \frac{1-\epsilon/2}{2}$ and $v^{\star} = (1-\epsilon/2)\OPT$, \algone \  returns $S$ that satisfies $\E\left[f(S)\right] \geq\frac{1 - \epsilon/2}{2e} \cdot (1-\epsilon/2)\OPT \geq \frac{1 - \epsilon}{2e} \cdot \OPT .$ For adaptivity, note that each iteration of \filter \ has two adaptive rounds: one for $\Delta(a, S, X)$ for all $a \in N$ and one for $\E_{R \sim \U(X)}\left[ f_S(R 
   	\cap X^+)\right]$.  Since $|X|$ decreases by a $1+ \epsilon/4$ fraction at every iteration of \filter, every call to \filter \ has  at most $\log_{1+\epsilon/4}(n)$ iterations. Finally,  as there are $r = 20 \epsilon^{-1} \log_{1+\epsilon/4}(n)$ iterations of \algone, the adaptivity is $\mathcal{O}\left(\log^2 n\right)$.
\end{proof}

\section{Experiments}
\label{sec:experiments}
Our goal in this section is to show that beyond its provable guarantees, \algone \ performs well in practice across a variety of application domains. Specifically, we are interested in showing that despite the fact that the parallel running time of our algorithm is smaller by several orders of magnitude than that of any known algorithm for maximizing non-monotone submodular functions under a cardinality constraint, the quality of its solutions are consistently competitive with or superior to those of state-of-the-art algorithms for this problem. To do so, we conduct two sets of experiments where the goal is to solve the problem of $\max_{S:|S|\leq k}f(S)$ given a function $f$ that is submodular and non-monotone. In the first set of experiments, we test our algorithm on the classic max-cut objective  evaluated on graphs generated by various random graph models. In the second set of experiments, we apply our algorithm to a max-cut objective on a new road network dataset, and we also benchmark it on the three objective functions and datasets used in \citep{mirzasoleiman2016fast}. In each set of experiments, we compare the quality of solutions found by our algorithm to those found by several alternative algorithms.

\subsection{Experiment set I: cardinality constrained max-cut on synthetic graphs}
%In the first set of experiments our goal was to observe the behavior of the algorithm against various objective function.    
%
%We conducted was by running the algorithm Max Cut objective function defined on four different random graphs.  For each random graph 
%
%In the cardinality-constrained Max Cut problem, we are given an undirected Boolean graph $G(V,E)$ and our goal is to find a set nodes $S$ of size at most $k$ such that the count of edges between $S$ and $V\backslash S$ is maximized. More precisely, we maximize:
%
Given an undirected graph $G=(N,E)$, recall that the cut induced by a set of nodes $S\subseteq N$ denoted $C(S)$ is the set of edges that have one end point in $S$ and another in $N\setminus S$. The cut function $f(S) = |C(S)|$ is a quintessential example of a non-monotone submodular function. To study the performance of our algorithm on different cut functions, we use four well-studied random graph models that yield cut functions with different properties. For each of these graphs, we run the algorithms from Section~\ref{sec:experiments_algs} to solve $\max_{S:|S|\leq k}|C(S)|$ for different $k$: % We briefly describe these graphs here and provide a detailed description of their construction in Appendix~\ref{sec:appendix_random_graphs}.

\begin{itemize}
\item \textbf{Erd\H{o}s R\'{e}nyi.}  We construct a $G(n,p)$ graph with $n=1000$ nodes, $p=1/2$, and use $k = 700$. Since each node's degree is drawn from a Binomial distribution, many nodes will have a similar marginal contribution to the cut function, and a random set $S$ may perform well.

\item \textbf{Stochastic block model.} We construct an SBM graph with 7 disconnected clusters of 30 to 120 nodes, a high $(p=0.8)$ probability of an edge within each cluster, and use $k = 360$.  Unlike for $G(n,p)$, here we expect a set $S$ to achieve high value only by covering all of the clusters.

\item \textbf{Barb\'{a}si-Albert.} We create a graph with $n=500$, $m=100$ edges added per iteration, and use $k = 333$. We expect that a relatively small number of nodes will have high degree in this model, so a set $S$ consisting of these nodes will have much greater value than a random set.

\item \textbf{Configuration model.} We generate a configuration model graph with $n=500$, a power law degree distribution with exponent $2$, and use $k = 333$. %This model results in a heavy-tailed network whose structure resembles many social networks. 
Although configuration model graphs are similar to {Barb\'{a}si-Albert} graphs, their high degree nodes are not connected to each other, and thus greedily adding these high degree nodes to $S$ is a good heuristic.  
\end{itemize}

%For every one of the generated graphs above we ran the algorithms described in Section~\ref{sec:experiments_algs} to solve $\max_{S:|S|\leq k}|C(S)|$ for different values of $k$.   

% Experiments II: Replicating ICML
\subsection{Experiment set II: performance benchmarks on real data}
To measure the performance of \algone \ on real data, we use it to optimize four different objective functions, each on a different dataset. Specifically, we consider a traffic monitoring application as well as three additional applications introduced and experimented with in \citep{mirzasoleiman2016fast}: image summarization, movie recommendation, and revenue maximization. We note that while these applications are sometimes modeled with monotone objectives, there are many advantages to using non-monotone objectives (see \citep{mirzasoleiman2016fast}). We briefly describe these objective functions and data here and provide additional details in Appendix~\ref{sec:appendix_real_data}.

%on image summarization, movie recommendation, and social network revenue maximization.  %To represent each of these objectives, we implement the objective functions introduced by \citet{mirzasoleiman2016fast}, which are nonmonotone submodular in the case of image summarization and movie recommendation and monotone submodular in the case of revenue maximization. 

%To benchmark the performance of \textbf{\textsc{[NIPS 2018]}} on real data, we implement the algorithms on four different data sets, each with a different objective function. Specifically, we first introduce a new dataset of vehicle counts in a large California highway network, which we use to maximize an objective on traffic and border monitoring. We then implement three additional experiments introduced by \citep{mirzasoleiman2016fast} on image summarization, movie recommendation, and social network revenue maximization. 

\begin{figure}[t]
\centering
\includegraphics[height=0.206\textwidth]{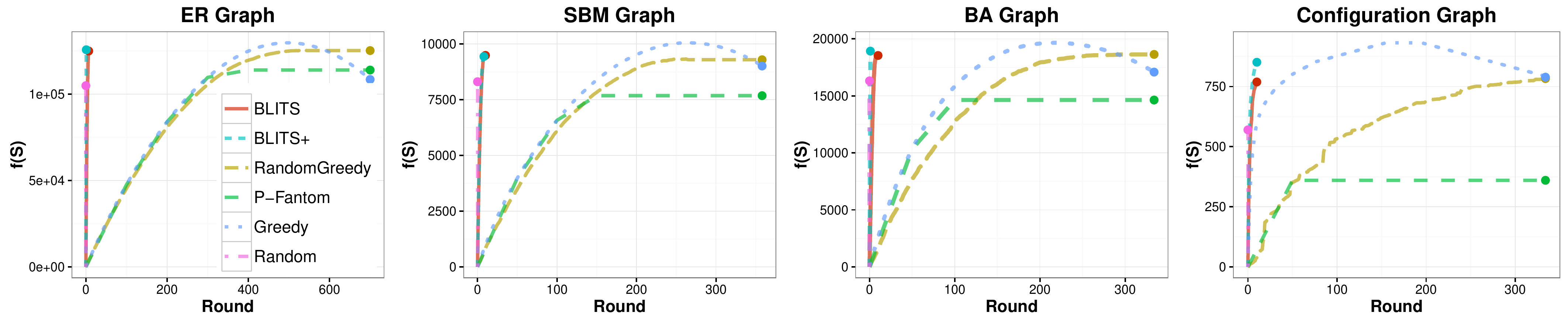}
\caption{\textit{Experiments Set 1: Random Graphs.} Performance of \algone \ (red) and \algtwo \ (blue) versus \textsc{RandomGreedy} (yellow), \textsc{P-Fantom} (green), \textsc{Greedy} (dark blue), and \textsc{Random} (purple).}
\label{fig:rand_graphs}
\end{figure}

\begin{itemize}
\label{ssec:real_dat_exp}

\item \textbf{Traffic monitoring.} Consider an application where a government has a budget to build a fixed set of monitoring locations to monitor the traffic that enters or exits a region via its transportation network. Here, the goal is not to monitor traffic circulating within the network, but rather to choose a set of locations (or nodes) such that the volume of traffic entering or exiting via this set is maximal. To accomplish this, we optimize a cut function defined on the weighted transportation network. More precisely, we seek to solve $\max_{S:|S|\leq k} f(S)$, where $f(S)$ is the sum of weighted edges (e.g. traffic counts between two points) that have one end point in $S$ and another in $N\setminus S$. To conduct an experiment for this application, we reconstruct California's highway transportation network using data from the CalTrans PeMS system~\citep{pems}, which provides real-time traffic counts at over 40,000 locations on California's highways, with $k = 300$. Appendix~\ref{sec:appendix_traffic} details on this network reconstruction. The result is a directed network in which nodes are locations along each direction of travel on each highway and edges are the total count of vehicles that passed between adjacent locations in April, 2018.

%\begin{align}
%\label{eqn:maxcutobj}
%f(S) = \sum_{i \in S} \sum_{j \in V\backslash S} A_{i,j} %-  \sum_{j \in S} \sum_{k \in S}A_{j,k}
%\end{align}
%where $A$ is the network's adjacency matrix. Note that this function is non-monotone, submodular, and nonnegative.\\
% %\ref{eqn:maxcutobj} 

%\item \textbf{Image summarization.} In this application we are given a large and diverse collection $X$ of images and we must select a small representative subset $S$. We conduct this experiment 500 randomly chosen images from the 10,000 Tiny Images dataset~\citep{tinyiamges}, and we measure how well one image represents another by the cosine similarity of their two raw RGB vectors;

\item \textbf{Image summarization.} Here we must select a subset to represent a large, diverse set of images. This experiment uses 500 randomly chosen images from the \emph{10K} Tiny Images dataset~\citep{tinyimages} with $k = 80$. We measure how well an image represents another by their cosine similarity.

%\item \textbf{Movie recommendation.} In this application our goal is to recommend a diverse short list $S$ of movies that are likely to be highly rated by a user based on the ratings she has assigned to movies she has already seen. We conduct this experiment on a randomly selected subset of 500 movies from the MovieLens dataset \citep{movielens}, which contains 1 million ratings by 6000 users on 4000 movies. Following \citep{mirzasoleiman2016fast}, we measure the similarity of one movie to another using the inner product of their respective columns of raw movie ratings;

\item \textbf{Movie recommendation.} Here our goal is to recommend a diverse short list $S$ of movies for a user based on her ratings of movies she has already seen. We conduct this experiment on a randomly selected subset of 500 movies from the MovieLens dataset \citep{movielens} of 1 million ratings by 6000 users on 4000 movies with $k = 200$. Following \citep{mirzasoleiman2016fast}, we define the similarity of one movie to another as the inner product of their raw movie ratings vectors.

%\item \textbf{Revenue maximization.} In this application one chooses a subset of $k$ users in a social network who will receive a product for free in exchange for advertising it to their network neighbors, and the goal is to choose these users in a manner that maximizes revenue. We conduct this experiment on a randomly selected subset of 25 communities ($\sim 1000$ nodes) from the 5000 largest communities in the Youtube social network dataset~\citep{youtube}, and we randomly assign each edge a weight by drawing from the uniform distribution $\U (0,1)$.

\item \textbf{Revenue maximization.} Here we choose a subset of $k = 100$ users in a social network to receive a product for free in exchange for advertising it to their network neighbors, and the goal is to choose users in a manner that maximizes revenue. We conduct this experiment on 25 randomly selected communities ($\sim$1000 nodes) from the 5000 largest communities in the YouTube social network ~\citep{youtube}, and we randomly assign edge weights from $\U (0,1)$.

\end{itemize}

\begin{figure}[t]
\centering
\includegraphics[height=0.206\textwidth]{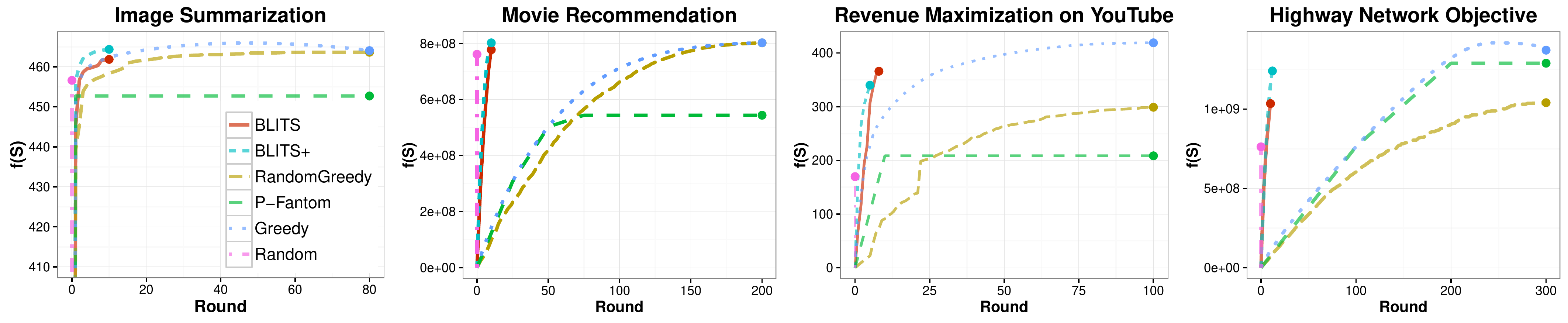}
\caption{\textit{Experiments Set 2: Real Data.} Performance of \algone \  (red) and \algtwo \  (blue) versus \textsc{RandomGreedy} (yellow), \textsc{P-Fantom} (green), \textsc{Greedy} (dark blue), and \textsc{Random} (purple).}
\label{fig:real_data}
\end{figure}

\subsection{Algorithms}\label{sec:experiments_algs} 
We implement a version of \algone \ exactly as described in this paper as well as a slightly modified heuristic, \algtwo .  The only difference is that whenever a round of samples has marginal value exceeding the threshold, \algtwo \ adds the highest marginal value sample to its solution instead of a randomly chosen sample.  \algtwo \  does not have any approximation guarantees but slightly outperforms \algone \ in practice.  We compare these algorithms to several benchmarks:

%In both sets of experiments we compared the performance of \textbf{\textsc{[NIPS 2018]}} and \textbf{\textsc{[NIPS++]}} to two algorithms with approximation guarantees for non-monotone objectives, \textsc{RandomGreedy}, and \textsc{p-Fantom}, and two additional `heuristic' algorithms that provide no such guarantees: \textsc{Greedy}, and \textsc{RandomK}. 
  
%
\begin{itemize}
\item \textbf{RandomGreedy}. This algorithm adds an element chosen u.a.r. from the $k$ elements with the greatest marginal contribution to $f(S)$ at each round.  It is a $1/e$ approximation for non-monotone objectives and terminates in $k$ adaptive rounds \cite{Buch14}.
\item \textbf{P-Fantom}.  \textsc{P-Fantom} is a parallelized version of the \textsc{Fantom} algorithm in~\citep{mirzasoleiman2016fast}.  \textsc{Fantom} is the current state-of-the-art algorithm for non-monotone submodular objectives, and its main advantage is that it can maximize a non-monotone submodular function subject to a variety of intersecting constraints that are far more general than cardinality constraints. The parallel version, \textsc{P-Fantom}, requires $\O(k)$ rounds and gives a $1/6 - \epsilon$ approximation.
\end{itemize} 

%\textsc{P-Fantom} works by iteratively building $S$ by calling a version of \textsc{Greedy} in which the element with the highest marginal contribution is only added to $S$ if its marginal value exceeds one of $|X|$ cleverly chosen thresholds, then paring $S$ by testing whether some subset of $S$ achieves higher value than $S$ itself.

We also compare our algorithm to two reasonable heuristics:

\begin{itemize}
\item \textbf{Greedy.} \textsc{Greedy} iteratively adds the element with the greatest marginal contribution at each round.  It is $k$-adaptive and may perform arbitrarily poorly for non-monotone functions.
\item \textbf{Random.} This algorithm merely returns a randomly chosen set of $k$ elements. It performs arbitrarily poorly in the worst case but requires 0 adaptive rounds.
\end{itemize}

\subsection{Experimental results}
For each experiment, we analyze the value of the algorithms' solutions over successive rounds (Fig. \ref{fig:rand_graphs} and \ref{fig:real_data}). The results support four conclusions. First, \algone \ and/or \algtwo \  nearly always found solutions whose value matched or exceeded those of \textsc{Fantom} and \textsc{RandomGreedy}--- the two alternatives we consider that offer approximation guarantees for non-monotone objectives. This also implies that  \algone \  found solutions with value far exceeding its own approximation guarantee, which is less than that of \textsc{RandomGreedy}. Second, our algorithms also performed well against the top-performing algorithm --- \textsc{Greedy}. Note that \textsc{Greedy}'s solutions decrease in value after some number of rounds, as \textsc{Greedy} continues to add the element with the highest marginal contribution each round even when only negative elements remain. While \algone's  solutions were slightly eclipsed by the \emph{maximum} value found by \textsc{Greedy} in five of the eight experiments, our algorithms matched \textsc{Greedy} on Erd\H{o}s R\'{e}nyi graphs, image summarization, and movie recommendation. Third, our algorithms achieved these high values despite the fact that their solutions $S$ contained $\sim$10-15\% fewer than $k$ elements, as they removed negative elements before adding blocks to $S$ at each round. This means that they could have actually achieved even higher values in each experiment if we had allowed them to run until $|S|= k$ elements. Finally, we note that \algone \  achieved this performance in many fewer adaptive rounds than alternative algorithms. Here, it is also worth noting that for all experiments, we initialized \algone \  to use only $30$ samples of size $k/r$ per round --- far fewer than the theoretical requirement necessary to fulfill its approximation guarantee. We therefore conclude that in practice, \algone \ 's superior adaptivity does not come at a high price in terms of sample complexity.

%hopefully $\exists$ results.\\
%- consistently outperforms Fantom and RandomGreedy, which are only 2 algs with guarantees for non-monotone
%-consistently finds solutions with value far exceeding approximation guarantee
%
%
%-(maybe add?) in our experiment, results underline difficulty of border patrol...need a very large k to get the majority of the value in the optimal unconstrained solution. This may be one reason why many states, such as cali/caltrans, have invested in monitoring stations at tens of thousands of points along their road networks.
%

\newpage
\bibliographystyle{alpha}
 \bibliography{biblio}

\appendix

\newpage

\section{Additional Discussion on Parallel Computing and Depth}
\label{sec:pram}

In the PRAM model,  the notion of depth measures the parallel runtime of an algorithm and is closely related to the concept of adaptivity. The \emph{depth} of a PRAM algorithm is the number of parallel steps of this algorithm on a shared memory machine with any number of processors. In other words, it is the longest chain of dependencies of the algorithm, including operations which are not necessarily queries. The problem of designing low-depth algorithms is well-studied , e.g. \cite{blelloch1996programming, blelloch2011linear,  berger1989efficient, rajagopalan1998primal, blelloch1998fast, blelloch2012parallel}. Our positive results extend to the PRAM model with \algone \ having  $\tilde{\mathcal{O}}(\log^3 n\cdot d_f)$ depth, where $d_f$ is the depth required to evaluate the function on a set.  
 The operations that our algorithms performed at every round, which are  set union and set difference over an input of size at most quasilinear,  can all be executed by algorithms with logarithmic depth using treaps \cite{blelloch1998fast}.
 While the PRAM model assumes that the input is loaded in memory, we consider the value query model where the algorithm is given oracle access to a function of potentially exponential size. 

\section{Missing Analysis from Section~\ref{sec:algorithm1}}
\label{sec:appalgorithm1}

\lemmeta*
\begin{proof}
We show by induction that $\E\left[f(S_i)\right] \geq \frac{i \alpha}{r} \left(1 - \frac{1}{r}\right)^{i-1} v^{\star}$. Observe that
\begin{align*}
\E\left[f(S_i)\right] & = \E\left[f(S_{i-1})\right] + \E\left[f_{S_{i-1}}(T_i)\right] & \\
& \geq \E\left[f(S_{i-1})\right]  + \frac{\alpha}{r} \left(\left(1 - \frac{1}{r}\right)^{i-1}v^{\star} - f(S_{i-1})\right) \\
& \geq \left(1 - \frac{\alpha}{r} \right) \E\left[f(S_{i-1})\right]  + \frac{\alpha}{r} \left(1 - \frac{1}{r}\right)^{i-1}v^{\star}  & \text{inductive hypothesis}\\
& \geq \left(1 - \frac{\alpha}{r} \right) \frac{(i-1) \alpha}{r} \left(1 - \frac{1}{r}\right)^{i-2} v^{\star}  + \frac{\alpha}{r} \left(1 - \frac{1}{r}\right)^{i-1}v^{\star}  \\
& \geq \frac{i \alpha}{r} \left(1 - \frac{1}{r}\right)^{i-1} v^{\star} & \alpha \leq 1
\end{align*}
Thus, with $i = r$,
$$ \E\left[f(S_r)\right] =  \alpha \left(1 - \frac{1}{r}\right)^{r-1} v^{\star}  \geq \frac{\alpha}{e} v^{\star}.$$
\end{proof}

%\lempruning*

%Experiments details
\section{Experiments and Implementation Details}
Here we provide details regarding the data, objective functions, and implementations for our experiments.
\label{sec:appendix_real_data}

\subsection{California highway network experiment}
\label{sec:appendix_traffic}

% Experiments III: adam's california experiment
\textbf{Motivation.} Consider the following set of problems: a state authority must choose where to build freeway stations to measure the commodities that are imported or exported from the state; a country bordering a disease epidemic must decide which border entry points will receive extra equipment to monitor the health of those who wish to enter; a law enforcement agency is charged with deciding where to deploy roadblocks so as to maximize the chance of arresting a suspect fleeing to the border. These problems and many others share two aspects in common: First, success depends on choosing a set of locations (or nodes) on a transportation network such that the volume of traffic entering or exiting via this set is maximal. Note that this objective differs starkly from classic vertex cover objectives, as in our cases we are not interested in commodities/people circulating within the state, so the volume of traffic (edge weight) moving between the chosen set of monitoring locations provides no value. Second, government authorities often have fixed project budgets (or, in the case of law enforcement, a fixed number of officers) and monitoring equipment carries fixed costs regardless of where it is installed. Therefore, these problems are accurately modeled via a cardinality constraint on the number of monitoring locations (nodes). Therefore, we propose that these problems can be modeled by applying the the cardinality-constrained max-cut objective to a directed, weighted transportation network. Specifically, we seek to solve $\max_{S:|S|\leq k} f(S)$, where $f(S)$ is the sum of weighted edges (e.g. traffic counts between two points) that have one end point in $S$ and another in $N\setminus S$.

\begin{figure}
\centering
\includegraphics[height=.4\textwidth]{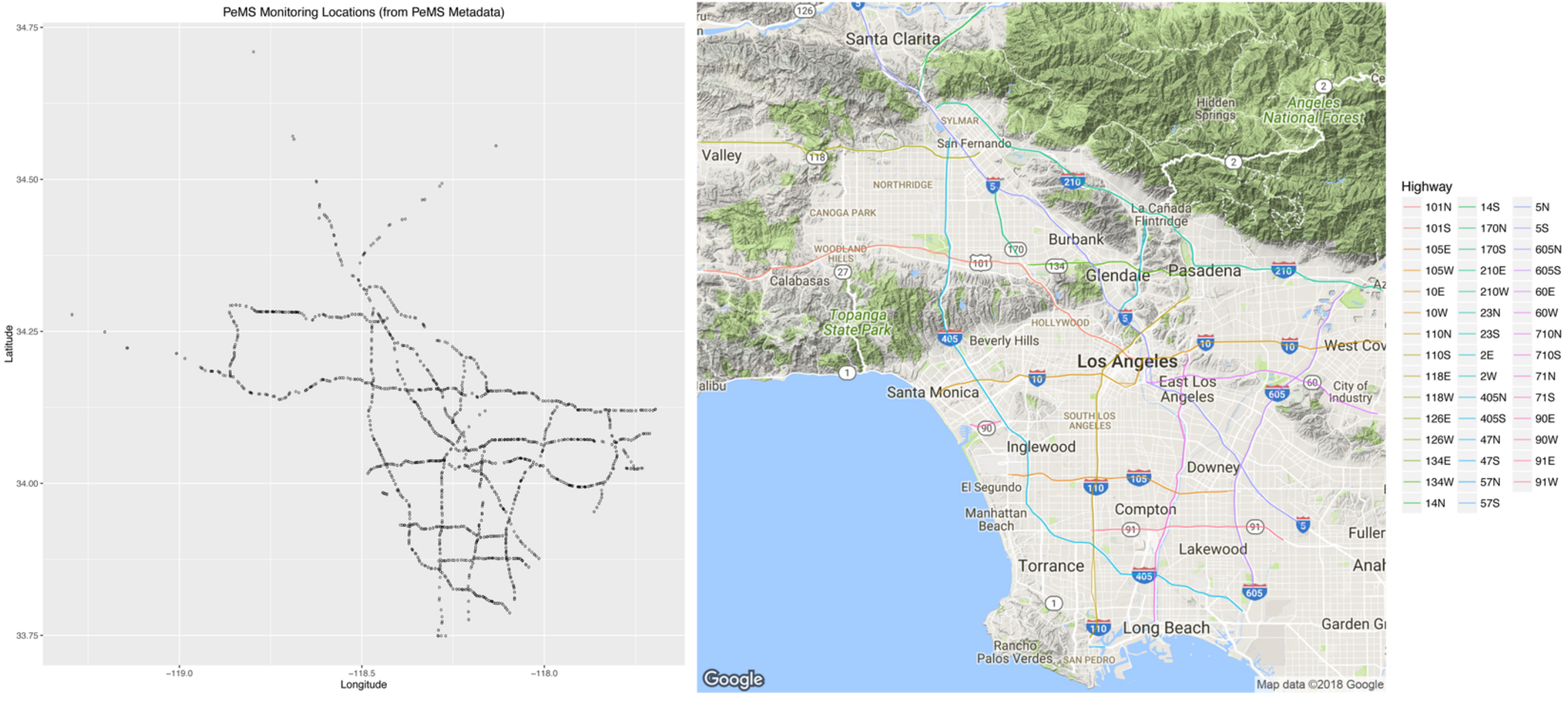}
\caption{(left) Raw metadata latitude and longitude plotted for $\sim$40,000 traffic counting monitors along CA highways; and (right) inferred highway network (each highway plotted in a different color).}
\label{fig:pemsraw}
\end{figure}

\begin{wrapfigure}{r}{0.38\textwidth}
\includegraphics[height=.38\textwidth]{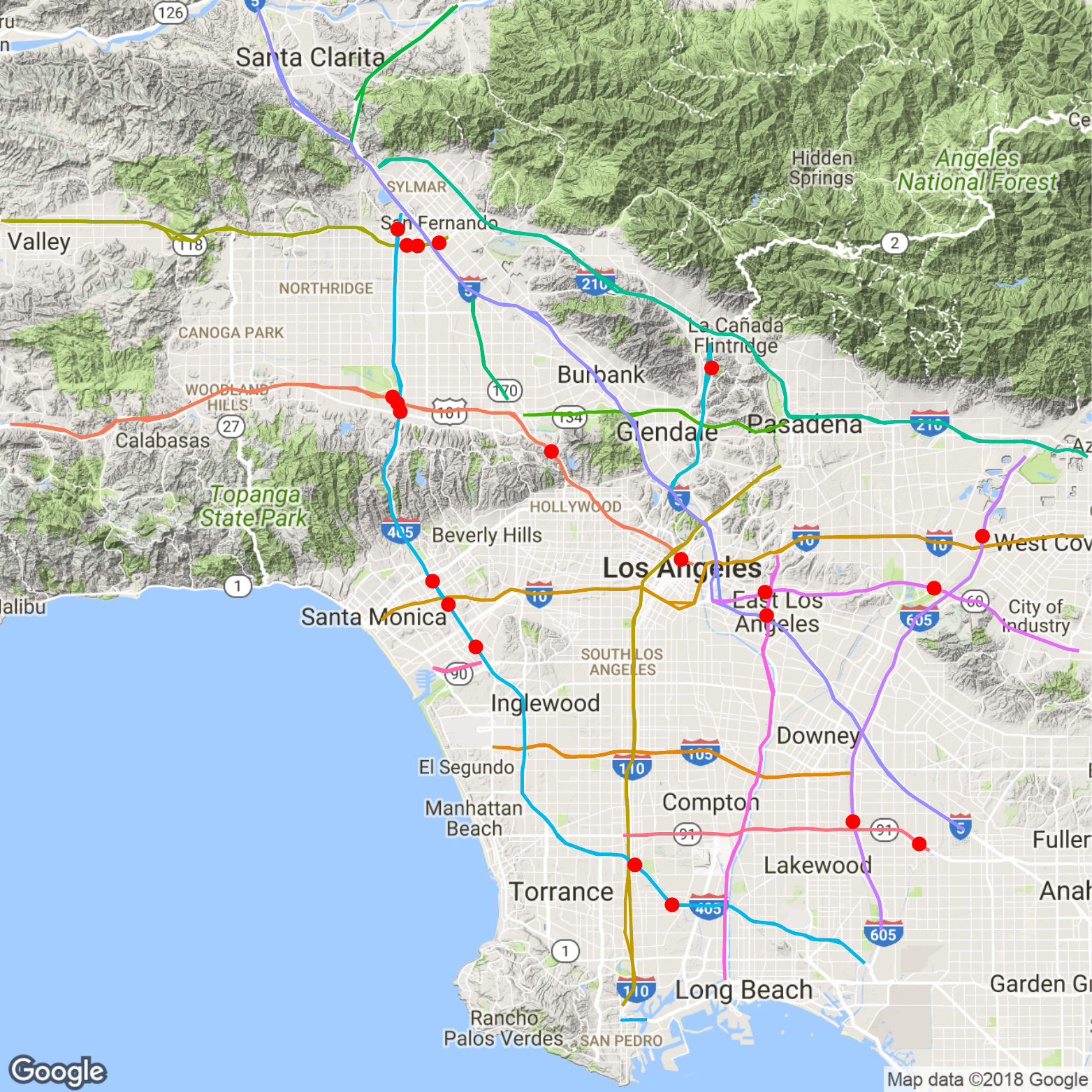}
\caption{The first 25 highway monitoring locations (red points) that \algone \ adds to its solution. Colored lines represent each of the 22 highways in the highway network inferred from metadata on the 3932 PeMs monitoring stations in Los Angeles and Ventura County.}
\label{fig:calisolutionmap}
\end{wrapfigure}

\textbf{Road network reconstruction.}
We reconstruct California's highway and freeway transportation network using data from the California Department of Transportation's (CalTrans) PeMs system \citep{pems}, which provides real-time traffic counts at over 40,000 locations along California's highways. Specifically, PeMs reports real-time traffic counts and metadata for each monitoring location (see Fig. \ref{fig:pemsraw}), but does not provide network information such as which stations form a sequence along each highway, and only a subset of monitoring stations adjacent to highway intersections are noted as such. Therefore, we cross-reference each PeMs traffic monitoring station's latitude, longitude, highway, and directional metadata with the Google Maps Location API to infer this network. The result is the network plotted on the right side of Fig. \ref{fig:pemsraw}. Specifically, nodes in this network are locations along each direction of travel on each highway and directed edges are the total count of vehicles that passed between adjacent locations for the month of April, 2018. Finally, we use the Google API to infer the location of missing edges representing highway intersections, and we impute these edges' respective edge weights using a simple linear model trained on the highway intersection traffic counts that are present in the data. In the spirit of our proposed applications, we restrict our network to the 22 highways comprising 3932 traffic monitors in LA and Ventura, and we restrict our solution to nodes within a 10mi radius of the Los Angeles network center. 

Fig. \ref{fig:calisolutionmap} plots the first 25 highway monitoring locations chosen by \algone \ against this highway network.

%
%
%\begin{figure}
%\centering
%\includegraphics[height=.3\textwidth]{"figures/pems_net_adam_inferred".pdf}
%%\includegraphics{"0ALL_Line_Plots".pdf}
%\caption{Inferred highway network (each highway is plotted in a different color)}
%\label{fig:pemsnetadam}
%\end{figure}

\subsection{Image, movie, and YouTube experiments}
%sec:appendix_random_graphs
% Image
\textbf{Image summarization experiment.}
In the image summarization application, we are given a large collection $X$ of images and we must select a small representative subset $S$. Following \citep{mirzasoleiman2016fast}, our goal is to choose our representative subset of images so that at least one image in this subset is similar to each image in the full collection, but the subset itself is diverse. These concerns inform the first and second terms of the non-monotone submodular objective function they propose:

\begin{align}
\label{eqn:imageobj}
f(S) = \sum_{i \in X} \max_{j \in S} s_{i,j} -  \frac{1}{|X|}\sum_{j \in S} \sum_{k \in S}s_{j,k}
\end{align}
where $s_{i,j}$ represents the cosine similarity of image $i$ to image $j$.

\textbf{Image Data.} As in \citep{mirzasoleiman2016fast}, we maximize this objective function on a randomly selected collection of 500 images from the Tiny Images `test set' data  \citep{tinyimages}, where each image is a 32 by 32 pixel RGB image.

% Movie
\textbf{Movie recommendation experiment.}
The goal of a movie recommendation system is to recommend a diverse short list $S$ of movies that are likely to be highly rated by a user based on the ratings she has assigned to movies she has already seen. \citep{mirzasoleiman2016fast} propose that this goal can be translated into the following non-monotone submodular objective function:

\begin{align}
\label{eqn:movie}
f(S) = \sum_{i \in S} \sum_{j \in X} s_{i,j} -  0.95 \sum_{j \in S} \sum_{k \in S}s_{j,k}
\end{align}
where $X$ is the set of all movies and $s_{i,j}$ is a measure of the similarity between movies $i$ and $j$.

\textbf{Movie data.} Following \citep{mirzasoleiman2016fast}, we optimize this objective function on a randomly selected set of 500 movies from the MovieLens 1M dataset \citep{movielens}, which contains 1 million ratings by 6000 users on 4000 movies. Because each user has rated only a small subset of the movies, we adopt the standard approach and use low-rank matrix completion to infer each user's rating for movies she has not seen in a manner that is consistent with her observed ratings. As in \citep{mirzasoleiman2016fast}, we use this completed ratings matrix to compute the movie similarity measure $s_{i,j}$ by setting $s_{i,j}$ equal to the inner product of the column of raw movie ratings of movies $i$ and $j$. %Finally, we select a random subset of $|X|=500$.
%Specifically, we impute these missing ratings using an implementation of the SVDimpute algorithm in the python package fancyimpute.

% Revenue Max Youtube
\textbf{Revenue maximization experiment.}
\citep{mirzasoleiman2016fast} also consider a variant of the influence maximization problem. Here, we can choose a subset of $k$ users of a social network who will receive a product for free in exchange for advertising it to their network neighbors, and the goal is to choose these users in a manner that maximizes revenue. More precisely, if we select the set $S$ of users to receive the product for free, then our revenue can be modeled via the following submodular objective function:

\begin{align}
\label{eqn:revenue}
f(S) = \sum_{i \in X \backslash S} \sqrt{ \sum_{j\in S} w_{i,j} }
\end{align}
where $X$ is the set of all users (nodes) in the network and $w_{i,j}$ is the network edge weight between users $i$ and $j$. Note that unlike the other objective functions, eqn. \ref{eqn:revenue} is monotone.

\textbf{Revenue maximization data.}
As in \citep{mirzasoleiman2016fast}, we conduct this experiment on social network data from the 5000 largest communities of the Youtube social network, which are comprised of 39,841 nodes and 224,234 undirected edges \citep{youtube}. Because edges in this data are unweighted, we randomly assign each edge a weight by drawing from the uniform distribution $U(0,1)$. We then optimize the revenue function on a randomly selected subset of 25 communities ($\sim 1000$ nodes).

\subsection{Implementation details}
The 8 plots in Fig. \ref{fig:rand_graphs} and \ref{fig:real_data} comparing the performance of \algone \ and \algtwo \ to alternatives show \emph{typical} performance. Specifically, with the exception of the plot depicting results for the movie recommendation and revenue maximization experiments, each plot was generated from a single run of each algorithm. For these runs, \algone \ and \algtwo \ were initialized with $r=10$, $\epsilon=0.3$, and $\OPT$ set to the value of the maximum value of a solution found by \textsc{Greedy}. This means that in practice, one could achieve higher values with \algone \ and \algtwo \ by running each algorithm multiple times in parallel (for which we do not incur an adaptivity cost) and picking the highest value $S$. For the movie recommendation and revenue maximization experiments, we noted that \algone \ was more sensitive to the $\OPT$ parameter. On the movie recommendation experiment, we noted a significant increase in the value of the solution $S$ returned by \algone \ as $\OPT$ was decreased from $f(S_{\textsc{Greedy}})$ to 0.7$f(S_{\textsc{Greedy}})$. On the revenue maximization experiment, we noted a significant increase in the value of the solution $S$ returned by \algone \ as $r$ was reduced to 5 and $\OPT$ was increased from $f(S_{\textsc{Greedy}})$ to 3.5$f(S_{\textsc{Greedy}})$, as this resulted in a more aggressive application of \filter. After exploring this behavior, we therefore set $\OPT$ to these better performing values and conducted one run of \algone \ each to produce the data for the movie recommendation and revenue maximization experiment plots.

Finally, we note that the plot lines for \textsc{P-Fantom} are less smooth than those for the other algorithms because unlike the other algorithms, \textsc{P-Fantom} does not construct a solution set $S$ only by adding elements to $S$. Specifically, \textsc{P-Fantom} works by iteratively building $S$ by calling a version of \textsc{Greedy} in which the element with the highest marginal contribution is only added to $S$ if its marginal value exceeds one of $|X|$ cleverly chosen thresholds, then paring $S$ by testing whether some subset of $S$ achieves higher value than $S$ itself. Therefore, unlike for the other algorithms we consider, running \textsc{P-Fantom} with a given constraint $k$ does not provide us with a single value of its partial solution in each round. Because it would be computationally costly to run \textsc{P-Fantom} for \emph{all} of these constraints (i.e. for all $\hat{k}\le k$) in order to generate values to plot against the other algorithms, we instead chose 10 equally spaced values in the interval $0< \hat{k}\le k$ and ran \textsc{P-Fantom} once for each.

\end{document}